\documentclass[graybox]{svmult}

\usepackage{type1cm}        
\usepackage{makeidx}         
\usepackage{graphicx}        
\usepackage{multicol}        
\usepackage{multirow}
\usepackage[bottom]{footmisc}

\usepackage{newtxtext}       %
\usepackage{newtxmath}       

\usepackage{latexsym,url,amscd,amsmath,amsfonts}
\usepackage{enumerate}
\usepackage{epstopdf}
\usepackage{array}
\usepackage[caption=false,font=normalsize,labelfont=sf,textfont=sf]{subfig}
\usepackage{textcomp}
\usepackage{stfloats}
\usepackage{url}
\usepackage{epsfig}
\usepackage{blindtext}
\usepackage{balance}
\usepackage{algpseudocode}
\usepackage{algorithm}
\usepackage{color}
\newcommand\numberthis{\addtocounter{equation}{1}\tag{\theequation}}

\newcommand{\norm}[1]{\left\lVert#1\right\rVert}

\newcommand{\R}{\mathbb{R}}

\newcommand{\vw}{\mathbf{w}}
\newcommand{\vx}{\mathbf{x}}

\newcommand{\cL}{\mathcal{L}}

\DeclareMathOperator*{\tr}{tr}

\DeclareMathOperator*{\sign}{sign}

\DeclareMathOperator{\shrink}{shrink}

\DeclareMathOperator*{\argmin}{argmin}

\DeclareMathOperator{\diag}{diag}
\DeclareMathOperator*{\vect}{vec}
\DeclareMathOperator*{\mat}{mat}
\DeclareMathOperator*{\erf}{erf}
\DeclareMathOperator*{\RMSE}{RMSE}
\providecommand{\keywords}[1]{\textbf{Keywords: } #1}

\makeindex

\begin{document}
\title*{Time-Varying Graph Signal Recovery Using High-Order Smoothness and Adaptive Low-rankness}
\titlerunning{Smooth and Low-Rank Time-Varying Graph Signal Recovery}
\author{Weihong Guo\and Yifei Lou \and Jing Qin \and Ming Yan}
\institute{
Weihong Guo
\at Department of Mathematics, Applied Mathematics, and Statistics\\
Case Western Reserve University\\
Cleveland, OH 44106\\
\email{wxg49@case.edu}
\and
Yifei Lou
\at  Department of Mathematics \& School of Data Science and Society\\
University of North Carolina at Chapel Hill\\
Chapel Hill, NC 27599\\
\email{yflou@unc.edu}
\and
Jing Qin \at Department of Mathematics\\
University of Kentucky\\
Lexington, KY 40506\\
 \email{jing.qin@uky.edu}
\and
Ming Yan \at School of Data Science\\
The Chinese University of Hong Kong, Shenzhen (CUHK-Shenzhen)\\
Shenzhen, China 518172\\
\email{yanming@cuhk.edu.cn}
}

\maketitle

\abstract{
Time-varying graph signal recovery has been widely used in many applications, including climate change, environmental hazard monitoring, and epidemic studies. It is crucial to choose appropriate regularizations to describe the characteristics of the underlying signals, such as the smoothness of the signal over the graph domain and the low-rank structure of the spatial-temporal signal modeled in a matrix form. As one of the most popular options, the graph Laplacian is commonly adopted in designing graph regularizations for reconstructing signals defined on a graph from partially observed data. In this work, we propose a time-varying graph signal recovery method based on the high-order Sobolev smoothness and an error-function weighted nuclear norm regularization to enforce the low-rankness. Two efficient algorithms based on the alternating direction method of multipliers and iterative reweighting are proposed, and convergence of one algorithm is shown in detail. We conduct various numerical experiments on synthetic and real-world data sets to demonstrate the proposed method's effectiveness compared to the state-of-the-art in graph signal recovery.
}


\section{Introduction}
Many real-world datasets are represented in the form of graphs, such as sea surface temperatures, Covid-19 cases at regional or global levels, and PM 2.5 levels in the atmosphere. Graphs play a crucial role in data science, facilitating the mathematical modeling of intricate relationships among data points. Typically composed of vertices with either undirected or directed edges, graphs regard each data point as a vertex and use edges to represent pairwise connections in terms of distances or similarities. A graph signal is a collection of values defined on the vertex set. The graph structure can be either provided by specific applications or learned from partial or complete datasets.

As an extension of (discrete) signal processing, graph signal processing \cite{shuman2013emerging} has become an emerging field in data science and attracted tremendous attention due to its capability of dealing with big data with irregular and complex graph structures from various applications, such as natural language processing~\cite{mills2013graph}, traffic prediction~\cite{tremblay2014graph}, climate change monitoring~\cite{sandryhaila2013discrete}, and epidemic prediction~\cite{giraldo2020minimization}. Graph signal recovery aims to recover a collection of signals with certain smoothness assumptions defined on a graph from partial and/or noisy observations. Unlike signals defined in traditional Euclidean spaces, the intricate geometry of the underlying graph domain must be considered when processing and recovering graph signals. Graph signals typically exhibit smoothness either locally or globally over the graph.

There are some challenges in graph signal recovery when exploiting the underlying graph structure to improve signal reconstruction accuracy. First, the topology of a graph desires a comprehensive representation involving many graph components, such as structural properties, connectivity patterns, vertex/edge density, and distribution. Second, it may be insufficient to describe the smoothness of graph signals by simply restricting the similarity of signal values locally. Moreover, the growth of graph size leads to a significant computational burden. To address them, various techniques have been developed, including graph-based regularization methods~\cite{chen2015signal, kroizer2019modeling,li2023robust, chen2024manifold}, spectral graph theory~\cite{varma2015spectrum,domingos2020graph,ozturk2021optimal,ying2021spectral}, and optimization algorithms~\cite{berger2017graph,jiang2020recovery}.

\subsection{Time-Varying Graph Signal Recovery}
A time-varying or spatial-temporal graph signal can be considered as a sequence of signals arranged chronologically, where each signal at a specific time instance is defined on a static or dynamically changing spatial graph.

Consider an undirected unweighted graph $G=(V, E),$ where $V$ is a set of $n$ vertices and $E$ is a set of edges. We assume a collection of time-varying graph signals $\{\vx_t\}_{t=1,\ldots,m}$ with $\vx_t\in\R^n$  are defined on $V$ with a time index $t$. Let $X=[\vx_1,\ldots,\vx_m]\in\R^{n\times m}$ be the data set represented in matrix. The pairwise connections on the graph $G$ can be modeled by an adjacency matrix $A$, where the $(i,j)$-th entry of $A$ is one if there is an edge between vertices $i$ and $j$, and zero otherwise. This binary adjacency matrix can be extended to the non-binary case for a weighted graph, where each entry indicates the similarity between two vertices. Throughout the paper, we use a standard $k$ nearest neighbor (kNN) approach based on the Euclidean distance of data points to construct the adjacency matrix.

Given an adjacency matrix $A$, we further define the graph Laplacian matrix, $L = M - A \in\mathbb{R}^{n\times n},$ where $M$ is a diagonal matrix with its diagonal element ${M}_{ii} = \sum_j {A}_{ij}$. The graph Laplacian serves as a matrix representation of the graph structure and can be used to describe some important characteristics of a graph, such as node connectivity and similarity. For example, geographic locations in the form of coordinates, i.e., longitude and latitude, are typically used to calculate the pairwise distance and, thereby, the graph Laplacian for geospatial data. For some data sets without obvious graph domains, a preprocessing step of graph learning can be implemented; see~\cite{xia2021graph} for a comprehensive review of graph learning techniques.

Time-varying graph signal recovery aims to recover an underlying matrix from its partially observed entries that are possibly polluted by additive noise. Mathematically, a forward model is $Y = J\circ X + \mathcal{N}$, where $Y$ is the observed data, $J \in\{0,1\}^{n\times m}$ is a sampling matrix, and $\mathcal{N}$ is a random noise. In this work, we focus on recovering time-varying signals, represented by the matrix $X$, from incomplete noisy data $Y$ defined on static spatial graphs in the sense that the vertex set and the edges do not change over time. In addition, we adopt a symmetrically normalized graph Laplacian that is pre-computed based on geographic locations.

\subsection{Related Works}\label{sec:prelim}
The recovery of graph signals from partial observations is an ill-posed problem due to missing data. Graph regularization plays a crucial role in developing a recovery model for time-varying signals by enforcing temporal correlation and/or describing the underlying graph topology.
An intuitive approach for recovering time-varying graph signals is to apply interpolation methods to fill in the missing entries, such as natural neighborhood interpolation (NNI)~\cite{sibson1981brief}. Numerous recovery models with diverse smoothness terms have been proposed to further preserve the underlying geometry. For example, Graph Smoothing (GS)~\cite{narang2013localized} characterizes the smoothness of the signal using the graph Laplacian of $X$. Alternatively, temporal smoothness is incorporated in Time-Varying Graph Signal Recovery (TGSR)~\cite{qiu2017time} by formulating the graph Laplacian of $DX,$ where $D$ is a first-order temporal difference operator. The combination of the
graph Laplacian of $X$ and the Tikhonov regularity of $DX$ was considered in~\cite{perraudin2017towards}. In contrast, the graph Laplacian of $DX$ with an additional low-rank regularity of $X$ was formulated as Low-Rank Differential Smoothness (LRDS)~\cite{mao2018spatio}. In the Tikhonov regularization, $\norm{XD}_F^2=\tr(XDD^TX^T)$ implies that $DD^T$ is treated as the temporal graph Laplacian. In~\cite{giraldo2022reconstruction}, the graph Laplacian matrix $L$ is replaced by $(L+\epsilon I)^{r}$, where $I$ is the identity matrix and $r\geq1$ for a high-order Sobolev spatial-temporal smoothness. Its main advantage lies in faster convergence, as this approach does not necessitate extensive eigenvalue decomposition or matrix inversion. Recently, another low-rank and graph-time smoothness (LRGTS) method has been proposed in~\cite{liu2023time}, where the sum of the nuclear norm and the Tikhonov regularizer on the second-order temporal smoothness is adopted to promote the low-rankness and the temporal smoothness, respectively.

All the models mentioned above can be condensed into one minimization framework:
\begin{equation}\label{eqn:univModel}
\begin{aligned}
\min_{X}~\frac12\norm{Y-J\circ X}_F^2+\frac{\alpha}2\tr(D_\theta^TX^TL_sXD_{\theta})+\beta R(X) + \frac{\gamma}2 \tr(XL_tX^T),
\end{aligned}
\end{equation}
where $D_{\theta}$ is a $\theta$-th order temporal difference operator, $L_s$ and $L_r$ are the spatial and temporal graph Laplacian matrices, respectively, $R(X)$ is the regularization term applied to $X$ describing its characteristics, and $\alpha\geq 0,~\beta\geq 0,~\gamma\geq 0$ are three parameters. Two common choices of $\theta$ are (1) $\theta = 0$ that corresponds to $D_\theta=I$ and (2) $\theta =1$ used in TGSR. Additionally, $L_s$ can be a transformed version of the classical graph Laplacian $L$, e.g., $\tilde{L}=(L+\epsilon I)^r$ in the Sobolev method \cite{giraldo2022reconstruction}, where $\epsilon>0$ and $r\geq 1$, which can be non-integer. The temporal graph Laplacian can be constructed by using the $\tau$-th order temporal difference operator, i.e., $L_t=D_\tau D_{\tau}^T$, for which case the temporal Laplacian can be expressed via the Frobenious norm $\tr(XD_{\tau}D_{\tau}^TX^T)=\norm{XD_{\tau}}_F^2$ (see Tikhonov with $\tau=1$ and LRGTS with $\tau=2$). The regularization $R(X)$ can be chosen as the nuclear norm of $X$ if the underlying time-varying graph signal $X$ is low rank.
Various models utilize different choices of $D_\theta, L_s/\tilde{L}, L_t$, and the regularization $R$. Leveraging the recent growth in deep learning, some time-varying graph signal recovery methods include unrolling technique \cite{kojima2023restoration}, graph neural network (GNN)~\cite{castro2023time}, and joint sampling and reconstruction of time-varying graph signals \cite{xiao2023joint}. In this work, we are dedicated to developing unsupervised time-varying graph signal recovery algorithms that do not involve or rely on data training.

\renewcommand{\arraystretch}{1.8}
\begin{table*}[ht]
\centering
\begin{tabular}{p{0.2\linewidth}p{0.8\linewidth}}
\hline
Method  & Optimization Model  \\
\hline
 GS \cite{narang2013localized} &
 $
\min_{X} \frac12\norm{Y-J\circ X}_F^2+\frac{\alpha}2 \tr(X^TLX)
$
($\theta=\beta =\gamma=0$ ) \\
\hline
 Tikhonov \cite{perraudin2017towards}   &
 ${\min_{X} \frac12\norm{Y-J\circ X}_F^2+ \frac{\alpha}2 \tr(X^TLX)+\frac{\gamma}2\norm{XD_1}_F^2}$ ($\theta=\beta = 0,L_t=D_1D_1^T$ )\\
\hline
TGSR \cite{qiu2017time}  &
$
\min_{X} \frac12\norm{Y-J\circ X}_F^2+\frac{\alpha}2\tr(D_1^TX^TLXD_1)
$ (
$\theta=1,\beta =\gamma = 0$) \\
\hline
LRDS \cite{mao2018spatio} &
$
\min_{X} \frac12\norm{Y-J\circ X}_F^2+\frac{\alpha}2\tr(D_1^TX^TLXD_1) +\beta \|X\|_*
$ \hspace{0.1in}($\theta=1,\gamma = 0$)\\
\hline

Sobolev \cite{giraldo2022reconstruction}  &
$
{\min_{X} \frac12\norm{Y-J\circ X}_F^2+\frac{\alpha}2\tr(D_1^TX^T(L+\epsilon I)^{r} XD_1)
}$ ($\theta=1,L_s=\tilde{L},\beta = \gamma= 0$) \\
\hline
LRGTS \cite{liu2023time} &
$
\min_X\frac12\norm{Y-J\circ X}_F^2+\frac{\alpha}2\tr(X^TLX)
+\beta\norm{X}_*+\frac{\gamma}2\norm{XD_2}_F^2$ ($\theta=0,L_t=D_2D_2^T$)
\\ \hline
Proposed L2 &
${\min_{X}\hspace{-2pt}\frac12\hspace{-2pt}\norm{Y-J\circ X}_F^2+\frac{\alpha}2\tr(D_{\theta}^TX^T(L+\epsilon I)^{r} XD_{\theta}) + \beta \norm{X}_{\text{erf}}}$ ($\gamma = 0$) \text{where} $ \norm{X}_{\text{erf}} $ \text{is an ERF weighted nuclear norm} \\
\hline

Proposed L1 &
$
\min_{X} \norm{Y-J\circ X}_1+\frac{\alpha}2\tr(D_{\theta}^TX^T(L+\epsilon I)^{r} XD_{\theta}) + \beta \norm{X}_{\text{erf}}
$ ($\gamma = 0$) \text{where} $ \norm{X}_{\text{erf}} $ \text{is an ERF weighted nuclear norm}   \\
\hline
\end{tabular}
\caption{Comparison of Related Works and Proposed Methods.}\label{tab1}
\end{table*}

Following the general framework \eqref{eqn:univModel}, we propose a novel low-rank regularization $R(X)$ based on the error function (ERF) \cite{guo2021novel} for sparse signal recovery (see Section~\ref{subsec:alg1}). In addition, to handle non-Gaussian type of noise such as Laplace noise, we propose a variant model in which the Frobeinus norm based data fidelity term is replaced with the $\ell_1$-norm data fidelity (see Section~\ref{subsec:alg2}). In Table \ref{tab1}, we provide a summary of the proposed models and relevant works pertaining to the general framework outlined in \eqref{eqn:univModel}.

\subsection{Contributions}\label{sec:contri}
The major contributions of this work are described as follows.
\begin{enumerate}
\item We develop a generalized time-varying graph signal recovery framework encompassing several state-of-the-art works as specific cases. We also develop two new models with a new regularization based on ERF.
\item The proposed models combine high-order temporal smoothness and graph structures with the temporal correlation exploited by iteratively reweighted nuclear norm regularization. 
\item We propose an efficient algorithm for solving the proposed models. Convergence analysis has shown that the algorithm generates a sequence that converges to a stationary point of the problem.
\item We conduct various numerical experiments, utilizing both synthetic and real-world datasets (specifically PM2.5 and sea surface temperature data), to validate the effectiveness of the proposed algorithm.
\end{enumerate}

\subsection{Organization}\label{sec:org}
The subsequent sections of this paper are structured as follows. In Section~\ref{sec:method}, we introduce a pioneering framework for recovering time-varying graph signals, leveraging Sobolev smoothness and ERF regularization. Additionally, we put forth an efficient algorithm based on the alternating direction method of multipliers (ADMM) and iterative reweighting scheme. A comprehensive convergence analysis of the proposed algorithm is also provided. In Section~\ref{sec:exp}, we present numerical experiments conducted on synthetic and real-world datasets sourced from environmental and epidemic contexts. Finally, Section~\ref{sec:con} encapsulates our conclusions and outlines potential avenues for future research.

\section{Proposed Method}
\label{sec:method}

\subsection{Error Function Weighted Nuclear Norm Regularization}

To enhance the low-rankness of a matrix, weighted nuclear norm minimization (WNNM) has been developed with promising performance in image denoising~\cite{gu2014weighted}. Specifically, the weighted nuclear norm (WNN) is defined as
\begin{equation}\label{eqn:matnorm}
\norm{L}_{\vw,*}:=\sum_{i}w_i\sigma_i(L),
\end{equation}
where $\sigma_i(L)$ is the $i$-th singular value of $L$ in the decreasing order and the weight vector $\vw=(w_i)$ is in the non-decreasing order with $w_i\geq0$ being the $i$-th weight. Choosing the weights is challenging in sparse and low-rank signal recovery problems. Iteratively reweighted L1 (IRL1)~\cite{candes2008enhancing} was proposed for the sparse recovery problem, where the weight is updated based on the previous estimate. It can solve many problems with complicated sparse regularizations, exhibiting improved sparsity and convergence speed.

In this work, we introduce a novel ERF-weighted nuclear norm based on the ERF regularizer~\cite{guo2021novel} and use linearization to obtain WNN. For any real matrix $X$ with $n$ singular values $\sigma_1(X)\geq \ldots\geq \sigma_n(X)$, the ERF-weighted nuclear norm is
\begin{equation}\label{eqn:erf} \norm{X}_{\erf}=\sum_{i=1}^n\int_0^{\sigma_i(X)}e^{-t^2/\sigma^2}dt, \end{equation} where $\sigma$ serves as a filtering parameter. To solve the ERF-nuclear norm regularized minimization problem, we use iterative reweighting (linearization) to get WNN with adaptive weights.

\subsection{Fractional-order derivative}
Inspired by the Gr\"{u}nwald-Letnikov fractional derivative~\cite{podlubny1999fractional}, we introduce the total $\theta$-th order temporal forward difference matrix with a zero boundary condition, as shown below
\begin{equation}\label{eq:frac-der}
    D_{\theta}=\begin{bmatrix}
C(0)&&&\\
\vdots&\ddots&&&\\
C(k)&\cdots&C(0)&&\\
&\ddots&&\ddots \\
&&C(k)&\cdots&C(0)
\end{bmatrix}\in\R^{m\times m}.
\end{equation}
Here the coefficients $\{C(i)\}_{i=0}^k$ are defined as
\[
C(i)=\frac{\Gamma(\theta+1)}{\Gamma(i+1)\Gamma(\theta+1-i)},\quad 0\leq i\leq k,
\]
where $\Gamma(x)$ is the Gamma function. Notice that if $\theta$ is a positive integer, $k$ can be deterministic. For example, if $\theta=1$, then $k=1$ and we have $C(0)=1$ and $C(1)=-1$, which is reduced to the first-order finite difference case. If $\theta=2$, then it reduces to the temporal Laplacian operator. Generally if $\theta=n$, then only the first $n+1$ coefficients $\{C(i)\}_{i=0}^n$ are nonzero and thereby $k=n+1$. For any fractional value $\theta$, we have to choose the parameter $k.$ The difference matrix \eqref{eq:frac-der} is built upon the zero boundary condition, while other types of boundary conditions, e.g., Newmann and periodic boundary conditions, can also be used. Alternatively, we can use low-order difference schemes for boundary conditions, e.g., the first-order forward difference based on the first $m-1$ time points and the zeroth order for the last time point.

\subsection{Proposed Algorithm 1}\label{subsec:alg1}
We propose the following ERF regularized time-varying graph signal recovery model
\begin{equation}\label{eqn:model1a}
\min_{X}~\frac12\norm{Y-J\circ X}_F^2+\frac{\alpha}2\tr(D_\theta^TX^T(L+\epsilon I)^{r}XD_\theta)+\beta \norm{X}_{\text{erf}}.
\end{equation}
Here we use the least squares as a data fidelity term, the Sobolev smoothness of time-varying graph signals \cite{giraldo2022reconstruction} as the graph regularization, and an ERF-based regularization defined in \eqref{eqn:erf} for temporal low-rank correlation.

We apply ADMM with linearization to solve the problem \eqref{eqn:model1a}. First, we introduce an auxiliary variable $Z$ to rewrite the problem~\eqref{eqn:model1a} into an equivalent constrained problem:
\[
\min_{X,Z}~\frac12\norm{Y-J\circ X}_F^2+\frac{\alpha}2\tr(D_\theta^TX^T(L+\epsilon I)^{r}XD_\theta)+\beta \norm{Z}_{\text{erf}},\text{ s.t. }X=Z.
\]

Since the proximal operator of $\|\cdot\|_{\text{erf}}$ is difficult to compute, we apply linearization on the ERF term to obtain a WNN when solving the subproblem for $Z$. The ADMM iterates as follows,
\begin{equation}
\begin{aligned}\label{eq:ADMM4alg1}
w_i\leftarrow & \exp(-\sigma_i^2(X)/\sigma^2),\qquad \text{for } i=1,\dots,m\\
Z\leftarrow &\argmin_{Z}~\beta\norm{Z}_{\vw,*}+\frac{\rho}2\norm{X-Z+\widehat{Z}}_F^2\\
X\leftarrow & \argmin_{X}\frac12\norm{J\circ X-Y}_F^2+
\frac{\alpha}2\tr(D_\theta^TX^T(L+\epsilon I)^{r}XD_\theta)+\frac{\rho}2\norm{X-Z+\widehat{Z}}_F^2\\
\hat Z \leftarrow &\hat Z + (X-Z),
\end{aligned}
\end{equation}
where $\rho>0$ is a stepsize that affects the convergence; please refer to Theorem \ref{thm: conv} for more details. We derive closed-form solutions for both $Z$- and $X$-subproblems in \eqref{eq:ADMM4alg1}.
Specifically for the $Z$-subproblem, it can be updated via the singular value thresholding operator, i.e.,
\begin{equation}\label{eqn:Z}
Z= SVT(X+\widehat{Z})=U\shrink(\Sigma,\diag(\beta\vw/\rho))V^T,
\end{equation}
where $U\Sigma V^T$ is the singular value decomposition of $X+\widehat{Z}$,  and $\diag(\cdot)$ is a diagonalization operator turning a vector into a diagonal matrix with the same entries as the vector. Here the shrink operator $\shrink(x,\xi)=\sign(x)*\max(|x|-\xi)$ is implemented entrywise, where $\sign(x)$ returns the sign of $x$ if $x\neq0$ and zero otherwise.

In the $X$-subproblem, we can rewrite the second term of the objective function as
\[\begin{aligned}
\tr(D_\theta^TX^T(L+\varepsilon I)^rXD_\theta)&=\norm{(L+\varepsilon I)^{r/2}XD_\theta}_F^2\\
&=\norm{(D_\theta^T\otimes (L+\varepsilon I)^{r/2})\vect(X)}_2^2
:=\norm{A\vect(X)}_2^2,
\end{aligned}\]
where $\otimes$ is the Kronecker product. Thus, the $X$-subproblem has the closed-form solution as
\begin{equation}\label{eqn:X}
X=\mat[(\widehat{J}+\alpha A^TA+\rho I)^{-1}(\widehat{J}^TY+\rho \vect(Z-\widehat{Z})))],
\end{equation}
where $\widehat{J}=\diag(\vect(J))$. Note that $\widehat{J}^TY=Y$ since $\widehat{J}$ is a diagonal matrix with binary entries in the diagonal, whose nonzero entries correspond to the sampled spatial points.
Furthermore, considering that the matrix $\widehat{J}+\alpha A^TA+\rho I$ is symmetric and positive definite, we perform its Cholesky factorization as $\widehat{J}+\alpha A^TA+\rho I=\tilde{L}\tilde{L}^T$. Subsequently, we leverage forward/backward substitution as a substitute for matrix inversion, thereby reducing computational time.
The pseudo-code of the proposed approach for minimizing the model \eqref{eqn:model1a} is given in Algorithm~\ref{alg1}.

\begin{algorithm}[ht]
\caption{Robust Time-Varying Graph Signal Recovery with High-Order Smoothness and Adaptive Low-Rankness}\label{alg1}
\begin{algorithmic}
\State\textbf{Input:} graph Laplacian $L$, parameters $\alpha,~\beta$, $\rho$, spatial Laplacian parameters $\epsilon$ and $r$, ERF parameter~$\sigma$, Fractional-order derivative parameters $\theta>0$ and integer $k\geq1$.
\State\textbf{Output:} $X$
\State\textbf{Initialize:} $X$, $\widehat{Z}$
\While{The stopping criteria is satisfied}
\State compute the weights $\vw$
\State update $Z$ via \eqref{eqn:Z}
\State update $X$ via \eqref{eqn:X}
\State $\widehat{Z}\leftarrow \widehat{Z}+(X-Z)$
\EndWhile
\end{algorithmic}
\end{algorithm}

\subsection{Proposed Algorithm 2}\label{subsec:alg2}
In real-world applications, the type of noise could be unknown, and it is possible to encounter a mixture of different types of noise. To enhance the robustness against noise, we propose the second model,
\begin{equation}\label{eqn:model2}
\min_X~\norm{Y-J\circ X}_1+\frac{\alpha}2\tr(D_\theta^TX^T(L+\epsilon I)^{r}XD_\theta)+\beta \norm{X}_{\text{erf}}
\end{equation}
Compared with \eqref{eqn:model1a}, this new model utilizes the $\ell_1$-norm data fidelity to accommodate various types of noise.
Because of the $\ell_1$ term, we need to introduce an additional variable $V$ to make the subproblems easy to solve. The equivalent constrained problem is

\[
\min_{\substack{J\circ X-Y=V\\X=Z}} \norm{
V}_1+\frac{\alpha}2\tr(D_\theta^TX^T(L+\epsilon I)^{r}XD_\theta)+\beta \norm{Z}_{\text{erf}}.
\]

Therefore, the ADMM with linearization on the ERF term has the following subproblems
\begin{equation}\begin{aligned}
V&\leftarrow \argmin_{V}\norm{V}_1+\frac{\rho_1}2\norm{J\circ X-Y-V+\widehat{V}}_F^2\\
Z&\leftarrow \argmin_{Z}\norm{Z}_{\vw,*}+\frac{\rho_2}2\norm{X-Z+\widehat{Z}}_F^2\\
X&\leftarrow \argmin_{X}
\frac{\alpha}2\tr(D_\theta^TX^T(L+\epsilon I)^{r}XD_\theta)+\frac{\rho_1}2\norm{J\circ X-Y-V+\widehat{V}}_F^2\\
&\qquad \qquad \phantom{\leftarrow}+\frac{\rho_2}2\norm{X-Z+\widehat{Z}}_F^2
\end{aligned}
\end{equation}
For the $V$-subproblem, we get the closed-form solution expressed via the shrinkage operator
\begin{equation}\label{eqn:V}
V=\shrink(\widehat{J}^T(Y+V-\widehat{V}),1/\rho_1).
\end{equation}
Similar to Algorithm 1, the solution of the $Z$-subproblem is given by \eqref{eqn:Z} with $\rho$ replaced by $\rho_2$. For the $X$-subproblem, we get the closed-form solution
\begin{equation}\label{eqn:X2}
X=\mat[(\rho_1\widehat{J}+\alpha A^TA+\rho_2 I)^{-1}(\rho_1\widehat{J}^T(Y+V-\widehat{V})+\rho_2 \vect(Z-\widehat{Z})))].
\end{equation}
The entire algorithm is described in Algorithm~\ref{alg2}.

\begin{algorithm}[ht]
\caption{Robust Time-Varying Graph Signal Recovery with High-Order Smoothness and Adaptive Low-Rankness}\label{alg2}
\begin{algorithmic}
\State\textbf{Input:} graph Laplacian $L$, parameters $\alpha,~\beta$, $\rho_1$, $\rho_2$, spatial Laplacian parameters $\epsilon$ and $r$, ERF parameter~$\sigma$, Fractional-order derivative parameters $\theta>0$ and integer $k\geq1$.
\State\textbf{Output:} $X$
\State\textbf{Initialize:} $X$, $\widehat{V}$, $\widehat{Z}$
\While{The stopping criteria is satisfied}
\State compute the weights $\vw$
\State update $V$ via \eqref{eqn:V}
\State update $Z$ via \eqref{eqn:Z}
\State update $X$ via \eqref{eqn:X2}
\State $\widehat{V}\leftarrow \widehat{V}+(J\circ X-Y-V)$
\State $\widehat{Z}\leftarrow \widehat{Z}+(X-Z)$
\EndWhile
\end{algorithmic}
\end{algorithm}

\subsection{Convergence Analysis of Algorithm~\ref{alg1}}

For simplicity, we define
$$f(X):=\frac12\norm{Y-J\circ X}_F^2+\frac{\alpha}2\tr(D_\theta^TX^T(L+\epsilon I)^{\gamma}XD_\theta)$$
and hence the augmented Lagrangian function is given by
$${\mathcal{L}}(X,Z,\hat Z)=f(X)+\beta \|Z\|_{\erf}+\rho\langle\hat Z,X-Z \rangle +{\frac{\rho}2}\|X-Z\|_F^2.$$
The function $f$ is convex and continuously differentiable. In addition, $\nabla f$ is Lipschitz continuous with a constant $L$. 

\begin{theorem} \label{thm: conv}
Let $\rho>L$ and $\{(X^k, Z^k,\hat Z^k)\}$ be a sequence generated from Algorithm~\ref{alg1}, then, the sequence is bounded and has a limit point that is a stationary point of the problem~\eqref{eqn:model1a}.
\end{theorem}

\begin{proof}
Consider one iteration of Algorithm~\ref{alg1}, the update of $Z^{k+1}$ gives
\begin{align*}
    \,&\cL(X^{k},Z^{k+1},\hat Z^{k})-\cL(X^k,Z^{k},\hat Z^k)\\
    =\,& \beta\|Z^{k+1}\|_{\text{erf}}+{\frac{\rho}2}\|X^k-Z^{k+1}+\hat Z^k\|_F^2-\beta\|Z^k\|_{\text{erf}} -{\frac{\rho}2}\|X^k-Z^k+\hat Z^k\|_F^2\\
    \leq \,& \beta\|Z^{k+1}\|_{w^k,*}-\beta\|Z^k\|_{w^k,*}+{\frac{\rho}2}\|X^k+\hat Z^k-Z^{k+1}\|_F^2-{\frac{\rho}2}\|X^k+\hat Z^k-Z^k\|_F^2\\
    \leq \,& -{\frac{\rho}2}\|Z^{k+1}-Z^k\|_F^2. \numberthis\label{AL1}
\end{align*}
The first inequality holds because the error function is concave for positive values. The second inequality is valid because $Z^{k+1}$ is the optimal solution of the $Z$-subproblem.

Then we consider the updates of $X^{k+1}$ and $\hat Z^{k+1}$, which together give
\begin{align*}
    &\cL(X^{k+1},Z^{k+1},\hat Z^{k+1})-\cL(X^k,Z^{k+1},\hat Z^k)\\
    =&f(X^{k+1})+\rho\langle\hat Z^{k+1},X^{k+1}-Z^{k+1}\rangle+{\frac{\rho}2}\|X^{k+1}-Z^{k+1}\|_F^2\\
    &-f(X^k)-\rho\langle \hat Z^k,X^k-Z^{k+1}\rangle-{\frac{\rho}2}\|X^k-Z^{k+1}\|_F^2\\
    =&f(X^{k+1})-f(X^k)+\rho\langle \hat Z^{k+1},X^{k+1}-X^k\rangle\\
    &+{\rho}\|\hat Z^{k+1}-\hat Z^{k}\|_F^2-{\frac{\rho}2}\|X^{k+1}-X^k\|_F^2,
\end{align*}
where the last equality uses the update $\hat Z^{k+1}=\hat Z^k+X^{k+1}-Z^{k+1}$. Since $f$ is smooth, the updates of $X^{k+1}$ and $\hat Z^{k+1}$ show that $\rho\hat Z^{k+1}+\nabla f(X^{k+1})=0$. The convexity and smoothness of $f$ give $f(X^{k+1})+\langle\nabla f(X^{k+1}),X^k-X^{k+1}\rangle+{\frac{1}{2L}}\|\nabla f(X^{k+1})-\nabla f(X^k)\|^2\leq f(X^k)$. Therefore, we have
\begin{align*}
    &\cL(X^{k+1},Z^{k+1},\hat Z^{k+1})-\cL(X^k,Z^{k+1},\hat Z^k)\\
    \leq & \left(\max\left({\frac1{\rho}}-{\frac1{2L}},0\right)L^2-{\frac{\rho}2}\right)\|X^{k+1}-X^k\|_F^2. \numberthis\label{AL2}
\end{align*}
If $\rho>L$, then $\max\left({\frac{1}{\rho}}-{\frac1{2L}},0\right)L^2-{\frac{\rho}2}<0$.

Combing the equations~\eqref{AL1} and~\eqref{AL2}, we see that $\cL(X^k,Z^k,\hat Z^k)$ is decreasing. Furthermore if $\rho>L$, we have
\begin{align*}
        &f(X^k)+\beta\|Z^k\|_{\text{erf}}+\rho\langle\hat Z^k,X^k-Z^k\rangle+{\frac{\rho}2}\|X^k-Z^k\|_F^2\\
    =   &f(X^k)+\beta\|Z^k\|_{\text{erf}}-\langle\nabla f(X^k),X^k-Z^k\rangle+{\frac{\rho}2}\|X^k-Z^k\|_F^2\\
    \geq& f(Z^k)+\beta\|Z^k\|_{\text{erf}}+{\frac{\rho-L}2}\|X^k-Z^k\|_F^2\geq 0, \numberthis\label{LB}
\end{align*}
where the last inequality comes from the Lipschitz continuity of $\nabla f$. So $\cL(X^k,Z^k,\hat Z^k)$ is bounded from below. Therefore, $\cL(X^k,Z^k,\hat Z^k)$ converges and
\begin{align}\lim_{k\rightarrow\infty} (X^{k+1}-X^k)=0,\quad \lim_{k\rightarrow \infty} (Z^{k+1}-Z^k)=0.\label{conxy}\end{align}
Since $\nabla f$ is Lipschitz continuous, we can get
\begin{align}\lim_{k\rightarrow\infty} \hat Z^{k+1}-\hat Z^k= X^{k}- Z^k=0.\label{conz}\end{align}

Next, we show that $(X^k,Z^k,\hat Z^k)$ is bounded. Since we have shown in~\eqref{LB} that
\begin{align*}{\mathcal{L}}(X^k,Z^k,\hat Z^k)
\geq & f(Z^k)+\beta \|Z^k\|_{\text{erf}} +{\frac{\rho-L}2}\|X^k-Z^k\|_F^2.
\end{align*}
Therefore, when $\rho>L$, the boundedness of ${\mathcal{L}}(X^k,Z^k,\hat Z^k)$ gives the boundedness of $f(Z^k)+\beta \|Z^k\|_{\erf}$ and $\|X^k-Z^k\|_F^2$. Thus, sequences $\{X^k\}$ and $\{Z^k\}$ are also bounded. Because $\rho \hat Z^k=-\nabla f(X^k)$, the sequence $\{\hat Z^k\}$ is also bounded.

Since the sequence $\{(X^k,Z^k,\hat Z^k)\}$ is bounded. There exists a convergent subsequence, that is, $(X^{k_i},Z^{k_i},\hat Z^{k_i})\rightarrow (X^\star,Z^\star,\hat Z^\star)$. The limits~\eqref{conxy} and~\eqref{conz} show that $(X^{k_i+1},Z^{k_i+1},\hat Z^{k_i+1})\rightarrow (X^\star,Z^\star,\hat Z^\star)$. Then we have that $X^\star=Z^\star$ and $\beta\partial \|Z^\star\|_{\erf}-\rho \hat Z^\star=0$. Thus, $X^\star$ is a stationary point of the original problem~\eqref{eqn:model1a}.
\end{proof}

\section{Numerical Experiments}\label{sec:exp}
In this section, we conduct various numerical experiments on synthetic and real data to demonstrate the performance of our proposed methods. In particular, we compare our methods - Algorithm 1 and Algorithm 2 - with other related states of the art, including natural neighbor interpolation (NNI) \cite{sibson1981brief}, graph smooth (GS) \cite{narang2013localized}, Tikhonov \cite{perraudin2017towards}, TGSR \cite{qiu2017time}, LRDS \cite{mao2018spatio}, and Sobolev \cite{giraldo2022reconstruction}. To evaluate the reconstruction quality, we adopt the root mean square error (RMSE) as a comparison metric, defined as follows
\begin{equation}\label{eq:RMSE}
   \RMSE=\frac{\|X-\widehat{X}\|_F}{\sqrt{nm}},
\end{equation}
where $\widehat{X}$ is the approximation of the ground truth graph signal $X\in\R^{n\times m}$ defined on a spatial-temporal graph with $n$ nodes and $m$ time instances.
All the numerical experiments are implemented on Matlab R2021a in a desktop computer with Intel CPU i9-9960X RAM 64GB and GPU Dual Nvidia Quadro RTX5000 with Windows 10 Pro.

\subsection{Synthetic Data}
Following the work of \cite{qiu2017time}, we generate $N=100$ nodes randomly from the uniform distribution in a $100\times 100$ square area. The graph weight is determined using $k$-nearest neighbors. Specifically, the weight between any two nodes is inversely proportional to the square of their Euclidean distance. We consider $k=5$ and visualize the corresponding graph in Fig.~\ref{fig:graph-con}.

\begin{figure}
    \centering
    \includegraphics[width=.65\textwidth]{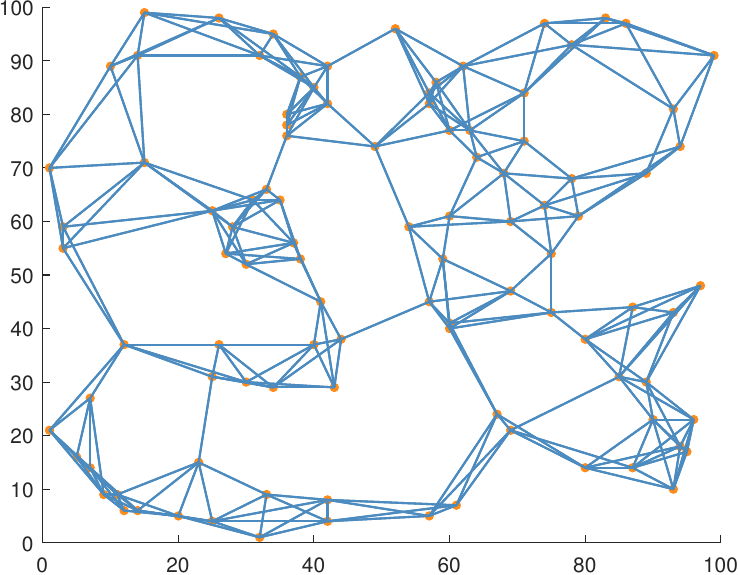}
    \caption{The graph is constructed by kNN with $k=5$. The weight between any two nodes is inversely proportional to the square of their Euclidean distance.}
    \label{fig:graph-con}
\end{figure}

Denote the weight matrix by $W,$ its degree matrix $D,$ and the graph Laplacian $L$ has eigen-decomposition $L=U\Lambda U^T,$ where $\Lambda =\text{diag}(0, \lambda_2, \cdots, \lambda_N)$.
We further define $L^{-1/2}=U\Lambda^{-1/2}U^T$ where $\Lambda^{-1/2} =\text{diag}(0, \lambda_2^{-1/2}, \cdots, \lambda_N^{-1/2})$. Starting from $x_1,$ we generate the time-varying graph signal
\begin{equation}\label{eq:time-varying-generate}
x_{t}=x_{t-1} + L^{-1/2} f_t, \quad \text{for } t = 2, \cdots, T,
\end{equation}
where $f_t$ is an i.i.d. Gaussian signal rescaled to $\|f_t\|_2=\kappa$ and $\kappa$ corresponds to a temporal smoothness of the signal. Stacking $\{x_t\}$ as a column vector, we obtain a data matrix $X=[x_1, x_2, \cdots, x_T]$. We generate a \textit{low-rank} data matrix obtained by starting with an empty matrix $X$ and repeating $X=[X, x_1, \cdots, x_{10}, x_{10}, x_9,\cdots, x_1]$ 10 times, thus also getting a $100\times 200$ data matrix. The measurement noise at each node is i.i.d. Gaussian noise $\mathcal N(0,\eta^2),$ where $\eta$ is the standard deviation.

\textit{Parameter tuning.} For the proposed Algorithm \ref{alg1}, we fix the following parameters: $k=3$ and $\theta=1.8$ in the definition of fractional-order derivative \eqref{eq:frac-der}; $\sigma=10^3$ in the definition of the ERF regularization \eqref{eqn:erf}; $\epsilon=0.1$ and $r=3$ in the Sobolev graph Laplacian; and the step size $\rho=10^{-6}$ in the ADMM iterations \eqref{eq:ADMM4alg1}. In each set of experiments, we carefully tune two parameters $(\alpha,\beta)$ that determine the weights for the spatial-temporal smoothness and the low-rankness, respectively,  in the proposed model \eqref{eqn:model1a}. We choose the best combination of $(\alpha,\beta)$ among $\alpha\in \{0, 10^{-5},10^{-4}, 10^{-3}, 10^{-2},10^{-1}, 1, 10\}$ and $\beta\in \{0, 10^{-8}, 10^{-7}, 10^{-6}, 10^{-5},10^{-4}, 10^{-3}, 10^{-2}, 10^{-1}, 1, 10\}$.
As demonstrated in Table~\ref{tab1}, some competing methods are special cases of the proposed models, and hence, we only tune the parameters $\alpha,\beta$ for these methods while keeping other parameters fixed.

\textit{Reconstruction errors with respect to sampling rates.} We begin by evaluating the performance of competing methods under different sampling rates. The smoothness level is set as $\kappa=1,$ while the standard deviation of the Gaussian noise is $\eta=0.1.$ The reconstruction performance is evaluated via RMSE, defined in \eqref{eq:RMSE}, showing that the recovery errors of all the methods decrease with the increase of the sampling rates. The comparison results are visualized in Fig.~\ref{fig:syn-exp-4b}. The proposed method achieves significant improvements over the competing methods. Surprisingly, LRDS, equipped with the nuclear norm, does not yield stable reconstruction performance in the low-rank case.

\begin{figure}
    \centering
    \includegraphics[width=0.65\textwidth]{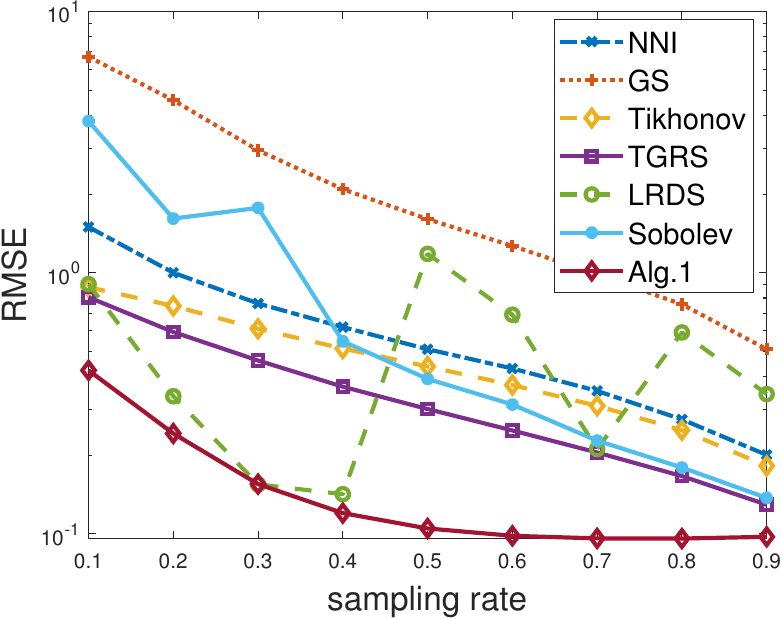}
    \caption{RMSE vs sampling rates. Averaged over 50 trials. }
    \label{fig:syn-exp-4b}
\end{figure}

\textit{Reconstruction errors with respect to noise levels.} We then investigate the recovery performance under different noise levels by setting the noise variance $\eta^2 = \{0.01, 0.1, 0.2, 0.4, 0.6, 0.8, 1\}.$ In this set of experiments, we fix the sampling rate as 40\% and smoothing level $\kappa=1.$ The noise level affects the magnitude of the least-squares fit, and as a result, we adjust the search window of $\alpha\in \{0, 10^{-3},10^{-2}, 10^{-1}, 1, 10, 10^{2},10^{3}, 10^4$. The parameter $\beta$ remains the same: $\beta\in \{0, 10^{-8}, 10^{-7}, 10^{-6}, 10^{-5},10^{-4}, 10^{-3}, 10^{-2}, 10^{-1}, 1, 10\}$. The results are presented in Fig.~\ref{fig:syn-exp-5a}, demonstrating the superior performance of the proposed Algorithm~\ref{alg1} under various noise levels.

\begin{figure}
    \centering
    \includegraphics[width=0.65\textwidth]{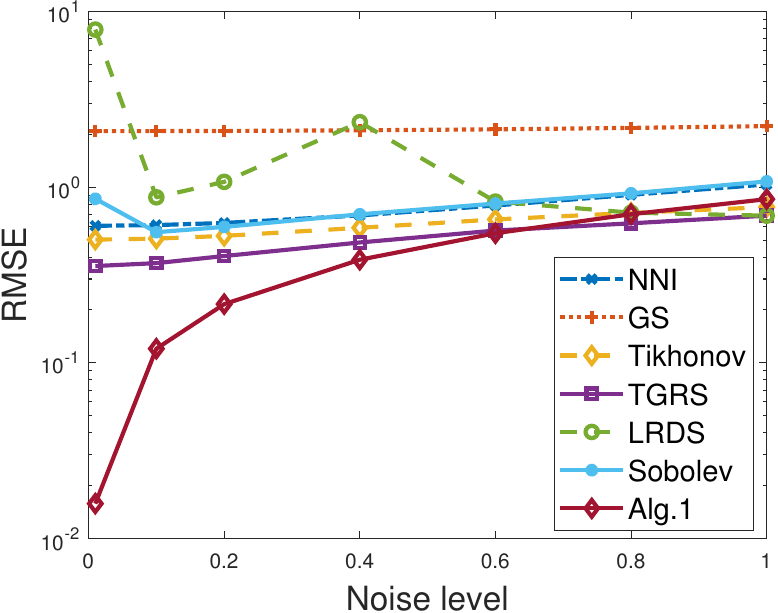}
    \caption{RMSE vs noise level: $\eta^2 = \{0.01, 0.1, 0.2, 0.4, 0.6, 0.8, 1\}.$ Averaged over 50 trials. }
    \label{fig:syn-exp-5a}
\end{figure}

\subsection{Real Data}

In the real data experiments, we first test the daily mean Particulate Matter (PM) 2.5 concentration dataset from California provided
by the US Environmental Protection Agency \url{ https://www.epa.gov/outdoor-air-quality-data}. We used the data captured daily from 93 sensors in California for the first 200 days in 2015. The constructed graph is depicted in Fig.~\ref{fig:PM25}. In Fig.~\ref{fig:pm2.5}, we compare the average recovery accuracy of all the comparing methods over 50 trials when the sampling rates are $0.1, 0.15, 0.2, 0.25, 0.3, 0.35, 0.4, 0.45$. In Table~\ref{tab:pm25_alg1v2}, we also compare the performance of Algorithm 1 and Algorithm 2, which shows Algorithm 2 can improve the accuracy of Algorithm 1 under some sampling rates with longer time in general.

\begin{figure}
    \centering    \includegraphics[width=.9\textwidth]
    {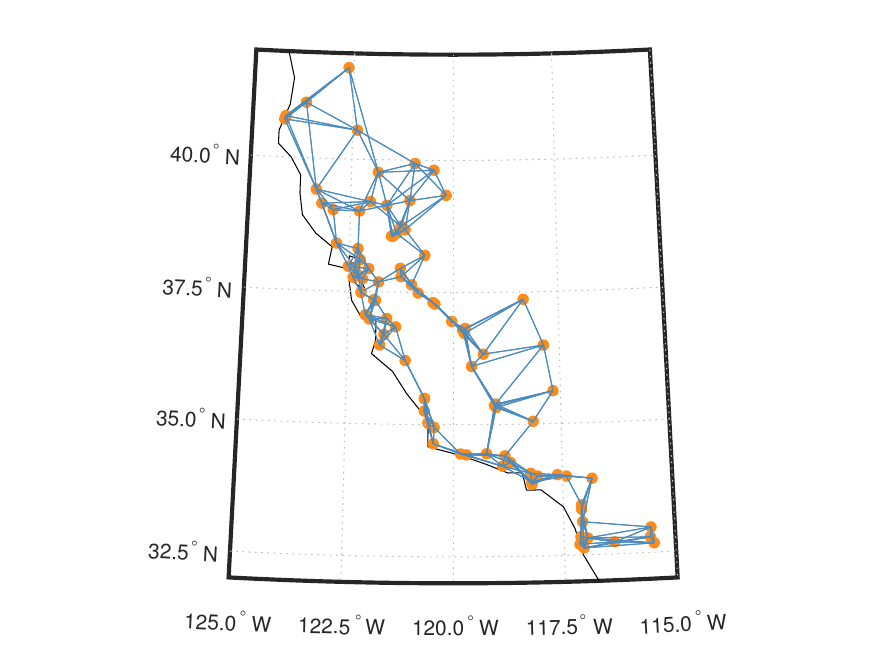}
    \caption{Graph with the places in California for the PM 2.5 concentration data. The graph was constructed with kNN for $k=5$.}\label{fig:PM25}
\end{figure}

\begin{figure}
\centering
\includegraphics[width=.65\textwidth]{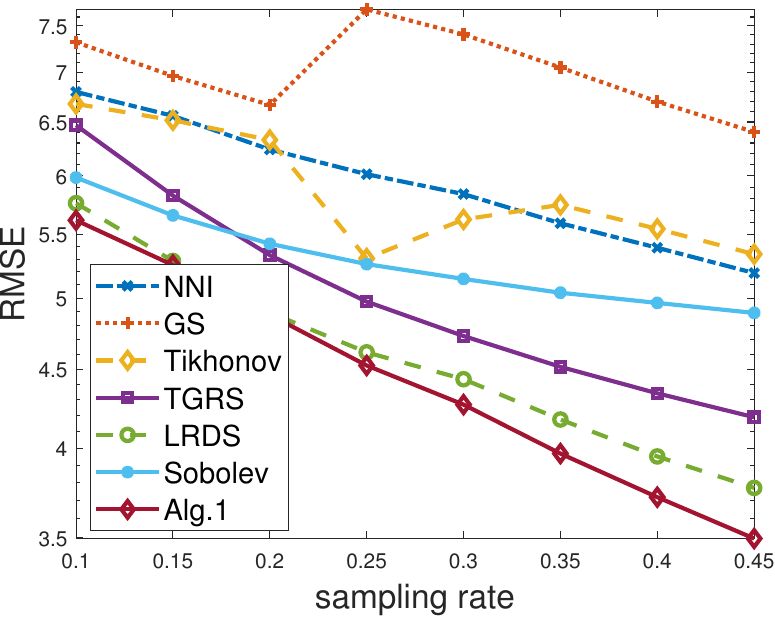}
\caption{Average recovery accuracy comparison on the PM2.5  data.}\label{fig:pm2.5}
\end{figure}

\begin{table}
\centering
\begin{tabular}{c|cc|cc}
\hline\hline
\multirow{2}{*}{Sampling rate} &\multicolumn{2}{c|}{Alg.1} & \multicolumn{2}{c}{Alg.2}\\ \cline{2-5}
&RMSE&Time (s)&RMSE&Time (s)\\ \hline
0.10&5.7321&22.95&5.6915&46.49\\
0.15&5.4770&22.63&6.4992&45.39\\
0.20&5.0427&23.83&5.9730&48.50\\
0.25&6.0358&23.37&5.6976&47.44\\
0.30&5.6065&23.70&5.3809&47.86\\
0.35&5.1920&23.55&5.1535&47.72\\
0.40&5.2398&23.56&4.7758&47.59\\
0.45&5.2283&23.80&5.0913&48.17\\ \hline\hline
\end{tabular}
\caption{Performance comparison of Algorithm 1 and Algorithm 2 for the PM2.5 data. The running time for Algorithm 1 is about 22$\sim$23 seconds while Algorithm 2 uses about $46\sim 48$ seconds.}\label{tab:pm25_alg1v2}
\end{table}

Next, we test the sea surface temperature dataset, which was captured monthly by the NOAA Physical
Sciences Laboratory (PSL). The data set can be downloaded from the PSL website \url{https://psl.noaa.gov/}. We use a subset of 200 time points on the Pacific Ocean
within 400 months. The constructed graph is illustrated in Fig.~\ref{fig:sea}. We see from Fig.~\ref{fig:seasurface} that the proposed algorithm outperforms other methods significantly and consistently across all sampling rates. In Table~\ref{tab:seasurf_alg1v2}, we also compare the performance of Algorithm~\ref{alg1} and Algorithm~\ref{alg2}, which indicates Algorithm~\ref{alg2} can improve the accuracy of Algorithm~\ref{alg1} under certain sampling rates but with more computational time in general.

\begin{figure}
    \centering    \includegraphics[width=.9\textwidth]
    {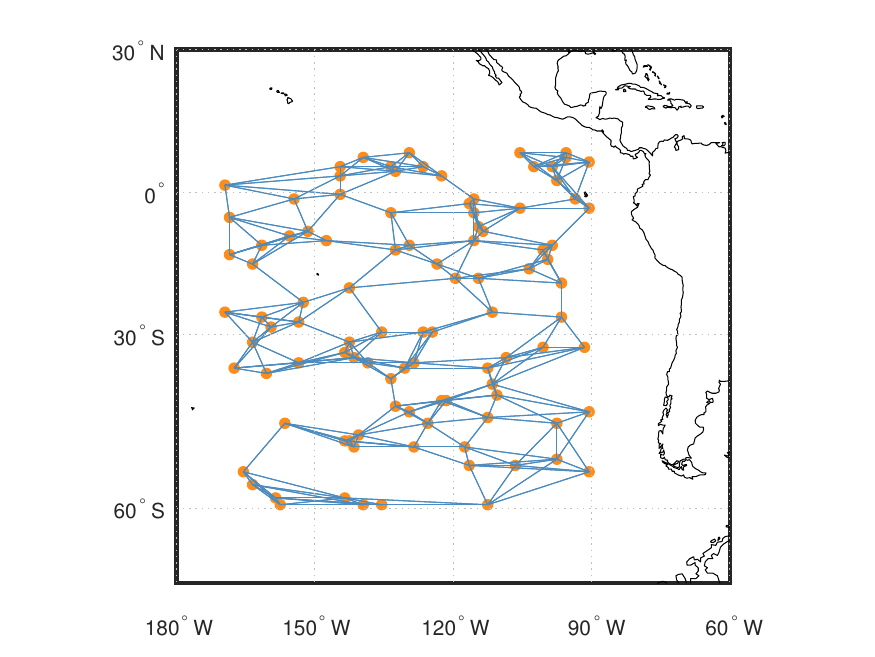}
    \caption{Graph with the places in the sea for the sea surface temperature data. The graph was constructed with kNN for
$k = 10$.}\label{fig:sea}
\end{figure}

\begin{figure}
\centering
\includegraphics[width=.65\textwidth]{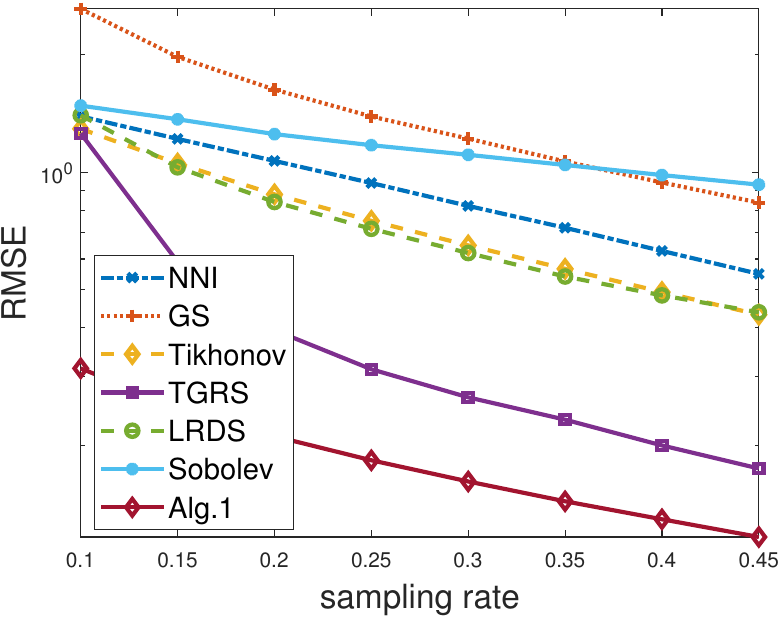}
\caption{Average recovery accuracy comparison on the sea surface temperature data.}\label{fig:seasurface}
\end{figure}

\begin{table}
\centering
\begin{tabular}{c|cc|cc}
\hline\hline
\multirow{2}{*}{Sampling rate} &\multicolumn{2}{c|}{Algorithm 1} & \multicolumn{2}{c}{Algorithm 2}\\ \cline{2-5}
&RMSE&Time (s)&RMSE&Time (s)\\ \hline
0.10&0.3148&3.97&0.3163&22.99\\
0.15&0.2497&3.37&0.2483&17.13\\
0.20&0.2110&3.08&0.2109&13.52\\
0.25&0.1832&2.87&0.1857&11.27\\
0.30&0.1617&2.74&0.1666&9.62\\
0.35&0.1438&2.63&0.1450&5.77\\
0.40&0.1294&2.54&0.1291&5.66\\
0.45&0.1166&2.46&0.1153&5.57\\ \hline\hline
\end{tabular}
\caption{Performance comparison of Algorithm 1 and Algorithm 2 for the sea surface temperature data. The running time for Algorithm 1 is about 2$\sim$4 seconds while Algorithm 2 uses about $6\sim 23$ seconds.}\label{tab:seasurf_alg1v2}
\end{table}

\subsection{Discussions}
Using the sea surface temperature data, we conduct an ablation study of the proposed model \eqref{eqn:model1a} without the smoothing regularization by setting $\alpha = 0$ or without the low-rank ERF term by setting $\beta =0$. We plot the RMSE curves with respect to the sampling rates and the noise levels in Fig.~\ref{fig:ablation}, showing that the ERF regularization has a larger influence on the performance compared to the Sobolev-base graph Laplacian regularization.

\begin{figure}
\centering
\includegraphics[width=.48\textwidth]{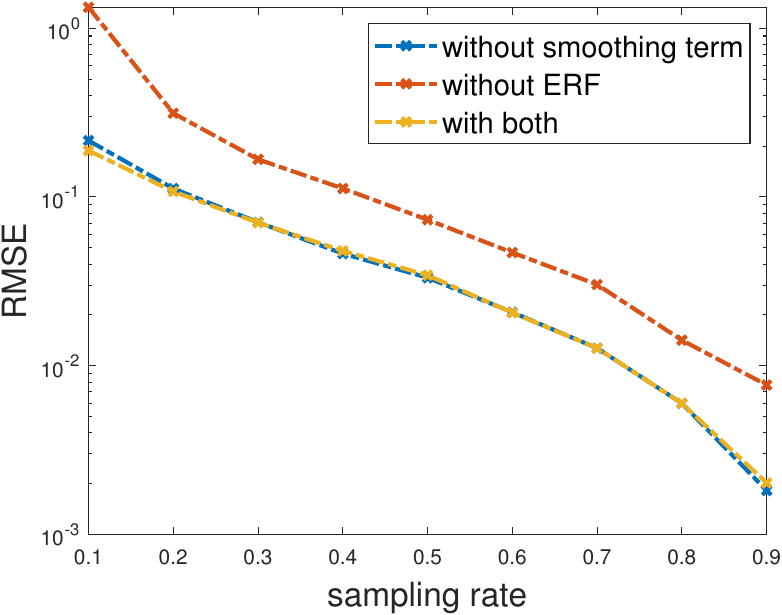}
\includegraphics[width=.48\textwidth]{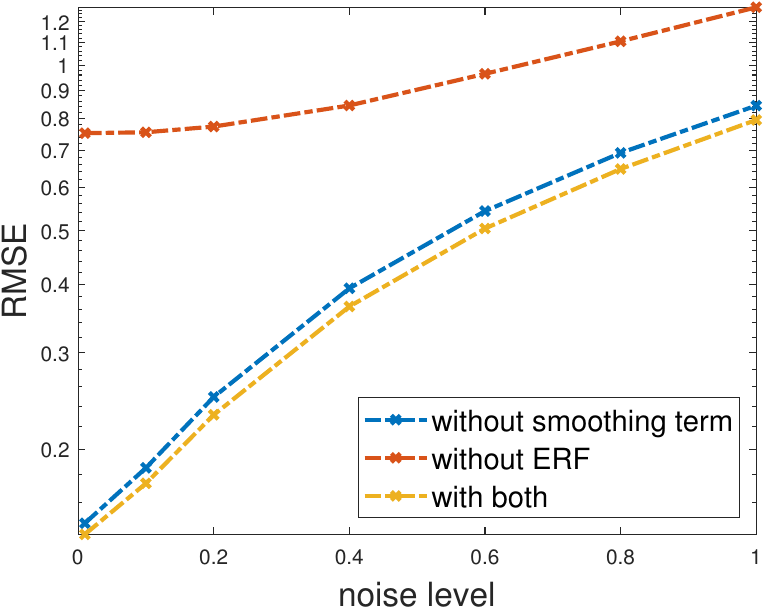}
\caption{Ablation study sampling rates (left) and noise levels (right) on the sea surface temperature data. }\label{fig:ablation}
\end{figure}

Using the same sea surface temperature data, we investigate whether the proposed model \eqref{eqn:model1a} is sensitive to the parameters $(r,\epsilon)$ in defining the Sobolev-graph Laplacian and $\sigma^2$ in defining the ERF regularization. Fig.~\ref{fig:sensitive}  shows that the proposed approach is not sensitive to various degrees of smoothness controlled by $r$ and $\epsilon$. Although the ERF regularization plays an important role in the recovery performance, as illustrated in the ablation study, the proposed model is not sensitive to the choice $\sigma^2$ as long as it is larger than 10,000.

In addition, we compare the proposed Algorithm~\ref{alg1} and Algorithm~\ref{alg2} using the sea surface temperature data and show the results in Tables \ref{tab:pm25_alg1v2} and \ref{tab:seasurf_alg1v2}. One can see that the two algorithms lead to similar RMSE, but Algorithm~\ref{alg2} is slower overall. We therefore prefer to use Algorithm~\ref{alg1} unless the data is heavily polluted by the non-Gaussian type of noise, such as Laplace noise.
\begin{figure}
\centering
\includegraphics[width=.48\textwidth]{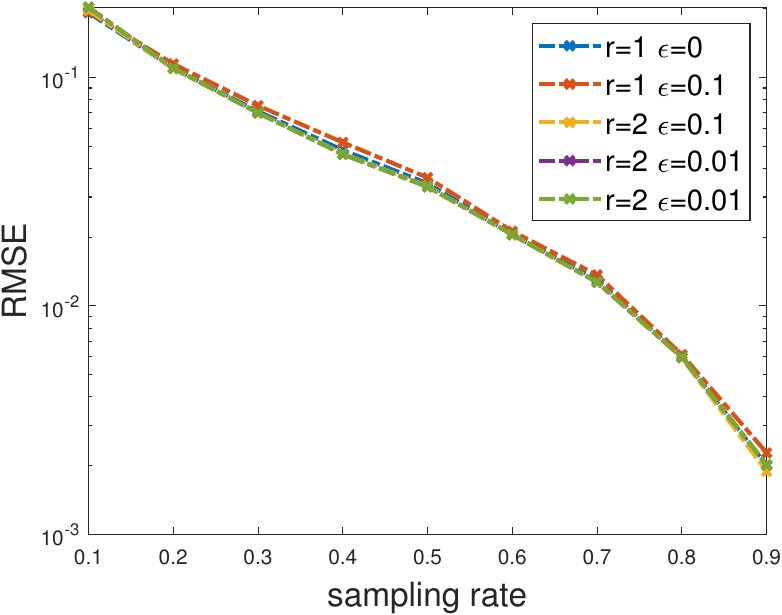}
\includegraphics[width=.48\textwidth]{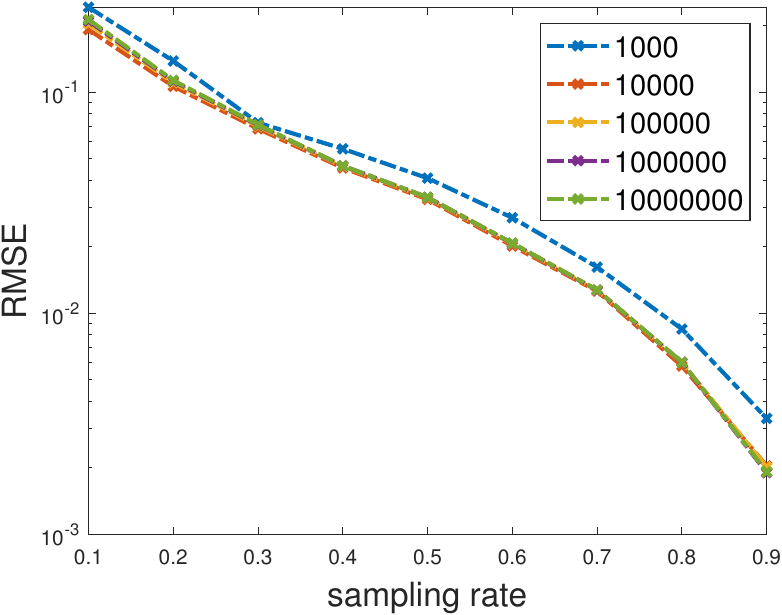}
\caption{Sensitivity analysis with respect to varying the graph Laplacian (left) and $\sigma^2$ in ERF (right) on the Sea Surface Temperature data. }\label{fig:sensitive}
\end{figure}

\section{Conclusions and Future Work}\label{sec:con}
In this paper, we exploit high-order smoothness across the temporal domain and adaptive low-rankness for time-varying graph signal recovery. In particular, we propose a novel graph signal recovery model based on a hybrid graph regularization involving a general order temporal difference, together with an error-function weighted nuclear norm. We also derive an effective optimization algorithm with guaranteed convergence by adopting a reweighting scheme and the ADMM framework. Numerical experiments have demonstrated their efficiency and performance in terms of accuracy. In the future, we will explore using high-order difference schemes to create a temporal Laplacian and low-rankness for recovering graph signals with dynamic graph topology.

\section*{Acknowledgements}
The authors would like to thank the support from the American Institute of Mathematics during 2019-2022 for making this collaboration happen. WG, YL, and JQ would also like to thank the Women in Data Science and Mathematics Research Workshop (WiSDM) hosted by UCLA in 2023 for the support of continuing this collaboration.  YL is partially supported by NSF CAREER 2414705.
JQ is partially supported by the NSF grant DMS-1941197. MY was partially supported by the Guangdong Key Lab of Mathematical Foundations for Artificial Intelligence, the Shenzhen Science and Technology Program ZDSYS20211021111415025, and the Shenzhen Stability Science Program.
\bibliographystyle{unsrt}
\bibliography{ref}

\end{document}